\documentclass[11pt]{article}
\usepackage{microtype}
\usepackage{mathrsfs}
\usepackage{amsthm,amsfonts,amscd,amsmath,amssymb}
\usepackage[margin=1in]{geometry}
\usepackage[english]{babel}
\usepackage{mathtools}
\usepackage{bbm}
\usepackage{libertine}
\usepackage{enumitem}
\usepackage{tikz}
\usetikzlibrary{backgrounds}
\usepackage{accents}
\usetikzlibrary{cd}
\usepackage[breaklinks=true,colorlinks=true,linkcolor=purple]{hyperref}
\hypersetup{
   citecolor=purple
}

\usepackage[style=numeric,
sortcites,
hyperref=true,
backend=bibtex]{biblatex}
\newcommand{\bbC}{\mathbb{C}}

\newcommand{\bbN}{\mathbb{N}}

\newcommand{\bbP}{\mathbb{P}}

\newcommand{\bbR}{\mathbb{R}}

\newcommand{\1}{\mathbbm{1}}
\DeclareMathOperator{\Tr}{tr}
\newcommand{\innerprod}[2]{\left\langle #1 , #2 \right\rangle}

\renewcommand{\vec}[1]{#1}

\DeclareMathOperator*{\expected}{\mathbb{E}}
\DeclareMathOperator{\avg}{avg}

\newcommand{\ER}{\text{Erd\H{o}s-R\'{e}nyi\,}}

\theoremstyle{plain}
  
  \newtheorem{corollary}{Corollary}
  \newtheorem{proposition}{Proposition}
  \newtheorem{lemma}{Lemma}
  \newtheorem{theorem}{Theorem}
  \newtheorem{question}{Question}
  
\theoremstyle{definition}

  \newtheorem{remark}{Remark}

\DeclareMathOperator{\Spec}{Spec}

\newcommand{\calG}{\mathcal{G}}

\newcommand{\calN}{\mathcal{N}}

\DeclareMathOperator{\girth}{girth}

\newcommand{\Lovasz}{Lov\`{a}sz\ }
\newcommand{\st}{\text{s.t.}}

\DeclareMathOperator{\dist}{dist}
\newcommand{\dir}{\accentset{\rightharpoonup}{E}}
\newcommand{\vcn}{\chi_v}
\DeclareMathOperator{\maxcut}{\textup{\textsc{MaxCut}}}
\DeclareMathOperator{\spr}{spr}
\newcommand{\ergraph}[2]{\calG({#1}, {#2})}
\newcommand{\ks}{\textsc{ks}}
\newcommand{\first}{\textsc{first}}
\newcommand{\second}{\textsc{second}}
\newcommand{\col}{\textsc{col}}

\begin{document}

\title{Vector Colorings of Random, Ramanujan, and Large-Girth Irregular Graphs}
\author{Jess Banks \thanks{Corresponding Author} \\ Dept. of Mathematics \\ University of California-Berkeley 
\and Luca Trevisan \\ Dept. of Computer Science \\ University of California-Berkeley}
\maketitle

\begin{abstract}
We prove that in sparse \ER graphs of average degree $d$, the vector chromatic number (the relaxation of chromatic number coming from the \Lovasz theta function) is typically $\tfrac{1}{2}\sqrt{d} + o_d(1)$. This fits with a long-standing conjecture that various refutation and hypothesis-testing problems concerning $k$-colorings of sparse \ER graphs become computationally intractable below the `Kesten-Stigum threshold' $d_{\ks,k} = (k-1)^2$. Along the way, we use the celebrated Ihara-Bass identity and a carefully constructed non-backtracking random walk to prove two deterministic results of independent interest: a lower bound on the vector chromatic number (and thus the chromatic number) using the spectrum of the non-backtracking walk matrix, and an upper bound dependent only on the girth and universal cover. Our upper bound may be equivalently viewed as a generalization of the Alon-Boppana theorem to irregular graphs.
\end{abstract}

\tableofcontents

\newpage

\section{Introduction}

Random graph coloring is one of the central and most studied problems in average case complexity, with over three decades of research interleaving the techniques and sensibilities  of theoretical computer science, statistical physics, and combinatorics. Many of the most striking phenomena occur in the case of sparse random graphs, and we will focus here on the \ER model $\ergraph{n}{d/n}$, where $d$ fixed and constant and each edge is included independently and with probability $d/n$. The full phenomenology   of this model is far beyond the scope of this paper to survey (we refer the reader to, for instance, \cite{zdeborova2007phase} for a more complete account), but its key aspect is a series of \emph{phase transitions} in the limit $n\to \infty$: for fixed $k$, there are critical thresholds in $d$ at which certain combinatorial and algorithmic attributes of the coloring problem change abruptly. 

The most famous of these is the \emph{colorability transition}, the threshold $d_{\col,k}$ below which graphs from $\ergraph{n}{d/n}$ are with high probability $k$-colorable (that is, with probability $1 - o_n(1)$ as $n\to \infty$), and above which they are not. Sophisticated refinements of the first and second moment methods \cite{achlioptas-naor,coja-oghlan-vilenchik,coja2013upper} have shown that
$$
  2k\log k - \log k - 1 + o_k(1) \triangleq d_{\first,k} \ge d_{\col,k} \ge d_{\second,k} \triangleq 2k\log k - \log k - 2\log 2 - o_k(1).
$$
These results pin down to within a small additive gap the threshold at which an exponential-time exhaustive search algorithm can find a coloring. What if, on the other hand, we care only about efficient algorithms, say those running in polynomial time? 

There are a number of algorithmic tasks that one can consider---distinguishing whether a graph was drawn from $\ergraph{n}{d/n}$ or from model with a `planted' $k$-coloring, finding exact or approximate colorings in graphs drawn from the latter, etc.---but all of them seem to become efficiently soluble only when
$$
  d > d_{\ks,k} \triangleq (k-1)^2;
$$
see \cite{massoulie2014,bordenave-lelarge-massoulie,abbe-sandon-more-groups,mns-colt,non-backtracking} for some examples, many of which are phrased in the related and more general case of \emph{community detection} which we do not treat here. It is conjectured that this point, known as the Kesten-Stigum threshold, is a universal barrier at which polynomial-time algorithms break down. 

The purpose of this paper is to add modest evidence to this conjecture, by studying a classic semidefinite programming algorithm for the problem of \emph{refutation}: given a graph $G \sim \calG(n,d/n)$, we are to efficiently produce a certificate that $G$ is not $k$-colorable or declare failure. As one cannot hope to refute $k$-colorability of $G$ when $d < d_{\second,k}$, the Kesten-Stigum threshold conjecture in our case asserts that when $d_{\first,k} < d < d_{\ks,k}$, refutation is possible but inaccessible to polynomial time algorithms, whereas it is efficiently soluble when $d_{\ks,k} < d$. e programming algorithm for refuting $k$-colorings.

To introduce our refutation algorithm, let us define a \emph{$k$-vector coloring} of an undirected graph $G=(V,E)$ as an assignment of a unit vector $v_i$ to each vertex $i\in V$, such that $\innerprod{v_i}{v_j} \leq -(k-1)^{-1}$ for every edge $(i,j) \in E$. This notion was introduced by Karger, Motwani, and Sudan in \cite{karger1998approximate}, and equivalent quantities date back to seminal works of \Lovasz and Schrijver \cite{lovasz1979shannon,schrijver1979comparison}. The vector chromatic number of $G$, which we will denote $\vcn(G)$, is the smallest $k$ (integer or otherwise) such that a $k$-vector coloring exists. If $G$ is $k$-colorable, then it is also $k$-vector-colorable (for instance by associating to each color one of the unit vectors pointing to the corners of a simplex in $\bbR^{k-1}$), so the vector chromatic number is a relaxation of the chromatic number. More importantly, it is a polynomial-time computable relaxation since it can be formulated as the following semidefinite program:

\begin{align}
	\label{chromatic-sdp}
	\vcn(G)
	= \min_P \, \kappa \qquad \st \qquad P &\succeq 0 \\
		P_{i,i} &= 1 & & \forall i \nonumber \\
		P_{i,j} &\le -(\kappa-1)^{-1} & & \forall (i,j) \in E \nonumber
\end{align}

A number of authors have studied the behavior of this and related semidefinite programs on sparse random graphs. In \cite{coja2005lovasz}, Coja-Oghlan shows concentration of the \Lovasz $\vartheta$ function for $G \sim \ergraph{n}{d/n}$, and an additional result that translates in our setting to $\vcn(G) = \Theta(\sqrt d)$, albeit with non-optimal constants. Montanari and Sen in \cite{montanari16} study an semidefinite programming algorithm for the problem of distinguishing $\ergraph{n}{d/n}$ from a planted model guaranteed to have a coloring or community structure, calculating its likely value up to an additive $o_d(1)$; the SDP that they consider is similar but incomparable with ours, as they are not concerned with refutation.

Our main theorem characterizes the vector chromatic number of sparse \ER graphs up asymptotically inconsequential terms as the average degree tends to infinity. This strengthens \cite{coja2005lovasz}, pinning down the constant exactly and substantially simplifying the method of proof.

\begin{theorem} \label{thm:main}
  When $G \sim \ergraph{n}{d/n}$, with probability $1 - o_n(1)$,
  \begin{equation*}
     \frac{d^{3/2}}{2d - 1} + 1 - o_n(1) \le \chi_v(G) \le \max\left\{\frac{d+1}{2\sqrt d} + 2, 4\right\}.
  \end{equation*} 
\end{theorem}

\noindent In other words, we determine that the threshold in $k$ below which the vector chromatic number can prove $G \sim \ergraph{n}{d/n}$ is not $k$-colorable, and above which it cannot, is $k = \tfrac{1}{2}\sqrt d + 1 + o_d(1)$. The careful reader will note that, although this matches the scaling of the Kesten-Stigum threshold, the constant factor out front is different: we have shown that refutation with the vector chromatic number becomes impossible when the average degree $d \gtrsim 4 d_{\ks,k}$. This shows that the conjectured ``hard regime'' $d_{\first,k} < d < d_{\ks,k}$ indeed stymies our refutation algorithm. Our result complements a result of Banks, Kleinberg, and Moore \cite{banks2019lovasz}, who have proved that in random $d$-regular graphs, $\chi_v(G)$ is similarly concentrated, and fails to refute $k$-coloring as well at four times that model's KS threshold. Together, these two papers raise a natural question: is this $4d_{\ks,k}$ scaling a fundamental barrier for efficient refutation, or can more elaborate methods (perhaps constantly many rounds of the Sum-of-Squares algorithm) succeed all the way down to the Kesten-Stigum threshold itself?

\section{Roadmap and Results} % (fold)
\label{sec:roadmap_and_results}

Banks et al. prove a lower bound on the vector chromatic number with a spectral argument, relying on Friedman's theorem \cite{friedman2003proof} to bound the smallest eigenvalue of the adjacency matrix of a random $d$-regular graph. The upper bound comes from an explicit construction of a feasible solution for the semidefinite program, using orthogonal polynomials. However, neither their upper nor lower bound extend to the $\ergraph{n}{d/n}$ model: the spectrum of the adjacency matrix is poorly behaved in \ER random graphs, and the use of orthogonal polynomials requires the graph to be regular. 

Instead, we will prove Theorem \ref{thm:main} by way of two deterministic results bounding the vector chromatic number of generic graphs. Both bounds are proved by way of non-backtracking walks. To state our results, let $G = (V,E)$ be an undirected graph on $|V| = n$ vertices, and denote by $A$, $D$, and $B$ its adjacency, diagonal degree, and non-backtracking matrices. We will introduce $B$ in detail below, but for now it is important only that it is a non-normal matrix with zero-one entries. Although its spectrum may be complex-valued, we verify in the sequel that the Perron-Frobenius theorem guarantees one real eigenvalue equal to the spectral radius, which we will denote $\spr(B) \triangleq \rho$. This quantity coincides with the growth rate of $G$'s universal covering tree, and its square root is the spectral radius of the non-backtracking operator on this infinite graph \cite{angel2015non,terras2010zeta}.

Our first deterministic result is that the spectrum of $B$ can certify non-colorability. 

\begin{theorem} \label{thm:lower}
  If $r$ is any lower bound on the smallest real eigenvalue of $B$, and $d_{\avg}$ is the average degree of $G$, then
  $$
      \chi_v(G) \ge \frac{|rd_{\avg}|}{r^2 + d_{\avg} - 1} + 1
  $$
\end{theorem}

\noindent To prove this lower bound, we use the celebrated Ihara-Bass identity (forthcoming in Theorem \ref{thm:ihara}) to relate the spectrum of $B$ to a family of symmetric matrices,
$$
  L(z) \triangleq z^2\1 - zA + D - \1 \qquad z\in\bbC
$$
known variously as the \emph{deformed Laplacian} or Bethe Hessian \cite{saade2014spectral,kotani2000zeta,angel2015non,bass1992ihara,hashimoto1989zeta}. It is observed in \cite[p.13]{fan2017well} that spectral assumptions on $B$ imply positive-definiteness of $L(z)$ for certain $z$ on the real line; we use these PSD matrices in a dual argument to lower bound $\chi_v(G)$. By a corollary of Bordenave et al. \cite{bordenave-lelarge-massoulie}, when $G\sim \ergraph{n}{d/n}$ we can with probability $1 - o_n(1)$ take $r \approx -\sqrt d$, giving the lower bound in Theorem \ref{thm:main}. 

Second, we derive a girth-dependent lower bound on $\vcn(G)$.

\begin{theorem} \label{thm:upper}
  If $\girth(G) \ge 2m + 1$, 
  $$
    \vcn(G) \le \frac{\rho + 1}{2(1 - 1/m)\sqrt\rho} + 1.
  $$
\end{theorem}

\noindent The feasible vector coloring we construct in the proof of Theorem \ref{thm:upper} assigns a $n$-dimensional unit vector $v_i$ to each vertex $i\in V$, whose coordinates we think of as again being indexed by $V$. In our construction, the coordinate $(v_i)_j$ is proportional to the \emph{square root of the probability of going from $i$ to $j$ in a certain non-backtracking random walk} of length equal to the distance between $i$ and $j$. This builds on the key idea in Srivastava and Trevisan's lower bound results for spectral sparsification \cite{srivastava2018alon}, and in the $d$-regular case recovers the result from Banks et al. \cite{banks2019lovasz}. Graphs drawn from $\ergraph{n}{d/n}$ have $\rho \approx d$, and this holds even if we condition on the constant probability event that the girth is any large constant of our choosing. Thus we can, with small albeit constant probability, construct $k$-vector colorings with $k$ arbitrarily close to $k = \tfrac{d+1}{2\sqrt{d}} + 1$. Finally, we adapt a well-known martingale technique developed in \cite{shamir1987sharp,luczak,achlioptas-moore-reg,banks2019lovasz} to guarantee, with high probability, a solution of similar cost.

The above construction can be used to prove two notable corollaries. First, it is also a near-optimal solution to the Goemans-Williamson relaxation of \textsc{MaxCut} in $\ergraph{n}{d/n}$ random graphs \cite{goemans1995improved}. Rounding with random hyperplanes yields a cut of cost 
$$
  |E| \cdot \left( \frac 12 + \frac {2-o_d(1) }{\pi} \cdot \frac 1 {\sqrt d} \right),
$$ 
which we believe is the strongest known algorithmically attainable lower bound to the maximum cut in $\ergraph{n}{d/n}$ random graphs (a tight bound is known, but the argument is not algorithmic \cite{dembo2017extremal}). In fact, this extends to any high-girth graph:

\begin{corollary} \label{cor:maxcut}
  If $\girth(G) \ge 2m + 1$,
  $$
    \maxcut(G) \ge |E|\left(\frac{1}{2} + \frac{2(1 - 1/m)\sqrt\rho}{\pi(\rho + 1)}\right).
  $$
\end{corollary}

Second, the vectors from Theorem \ref{thm:upper} can be used to prove a kind of generalized Alon-Boppana type theorem concerning the deformed Laplacian $L(z)$. The standard Alon-Boppana theorem \cite{nilli1991second} states that $d$-regular graphs with high diameter have have eigenvalues arbitrarily close to $2\sqrt{d-1}$; it has been refined and extended in numerous ways \cite[\S1.3-3]{davidoff2003elementary}\cite[\S3]{friedman1993some}\cite{nilli2004tight}, and our result generalizes the fact that regular graphs of large \emph{girth} have eigenvalues approaching $-2\sqrt{d-1}$. One can verify that these negative eigenvalues translate to eigenvalues of $L(z) = z^2\1 - zA + D - \1$ close to $(z + \sqrt{d-1})^2$ for every $z< 0$. For regular graphs $d-1$ is, among other things, the spectral radius of $B$, and we prove a direct generalization in this sense.

\begin{corollary} \label{cor:alon}
  If $G$ has girth at least $2m + 1$, then for every $z<0$,
  $$
    L(z) \not\succeq (z + \sqrt\rho)^2 - 2\sqrt\rho z/m.
  $$
\end{corollary}

We will prove Theorems \ref{thm:lower} and \ref{thm:upper} in \S\ref{sub:lower-pf}-\ref{sub:upper-pf} after first developing some preliminary results on non-backtracking walks in \S\ref{sub:preliminary_material}. Having done so, we prove Theorem \ref{thm:main} in \S\ref{sub:main-pf} and wrap up in \S\ref{sub:corollaries} with the two corollaries above.

\subsection{Optimality and Irregular Ramanujan Graphs} % (fold)
\label{sub:optimality_ramanujan_graphs_and_further_questions}

The best possible setting of $r$ in Theorem \ref{thm:lower} is $-\sqrt{d_{\avg} - 1}$, at which point we obtain the bound
$$
    \vcn(G) \ge \frac{d_{\avg}}{2\sqrt{d_{\avg} - 1}} + 1.
$$
In the case of $d$-regular Ramanujan graphs---those for which the nontrivial eigenvalues of the adjacency matrix have magnitude at most $2\sqrt{d-1}$---this matches the standard spectral bound on the chromatic number. For regular graphs, the Ramanujan property is euqivalent to every nontrivial eigenvalue of $B$ having magnitude at most $\sqrt{d-1}$; since $\rho = d-1$ in the regular case, some authors to define an irregular graph as Ramanujan if its nontrivial non-backtracking eigenvalues have modulus at most $\sqrt\rho$ \cite{bordenave-lelarge-massoulie,lubotzky1995cayley}. If a graph is Ramanujan in this sense, we can take $r = -\sqrt\rho$, giving
$$
  \vcn(G) \ge \frac{d_{\avg}\sqrt\rho}{\rho + d_{\avg} - 1} + 1;
$$
this could only match our upper bound in the case $\rho = d_{\avg} - 1$, which is true for regular graphs, approximately true for \ER random graphs, and fails generically.

\begin{question}
  What ``Ramanujan'' assumption on the spectrum of $B$ implies the converse of Theorem \ref{thm:upper}? Is it enjoyed, either approximately or exactly, by random graphs?
\end{question}

% subsection optimality_ramanujan_graphs_and_further_questions (end)

% section roadmap_and_results (end)

\section{Proofs} % (fold)
\label{sec:proofs}

\subsection{Notation and Non-backtracking Preliminaries} % (fold)
\label{sub:preliminary_material}

We will write $\Spec X$ for the unordered set of eigenvalues of a matrix $X$, $\spr X$ for the modulus of its largest eigenvalue, and use the standard notation $X \succeq 0$ to indicate that a (Hermitian) matrix is positive semidefinite, or in other words that $\Spec X \subset \bbR_{\ge 0}$. For two matrices $X$ and $Y$, $X\odot Y$ will denote the entry-wise product and  $\langle X,Y \rangle = \Tr YX^\ast = \sum_{i,j} \overline{X_{i,j}}Y_{i,j}$ the Frobenius inner product. It is a standard lemma that $X,Y \succeq 0$ implies $X\odot Y \succeq 0$ as well, and that $\langle X,Y\rangle \ge 0$. The set of integers $\{1,...,k\}$ will be denoted by $[k]$.

To an unweighted, undirected, and connected graph $G = (V,E)$ on $n$ vertices, we will associate an \emph{adjacency matrix} $A$, \emph{diagonal degree matrix} $D$, and shortest path distance metric $\dist : V \times V \to \bbN$. Although $G$ is undirected, it will be useful to think of each edge $(i,j) \in E$ as a pair of directed edges $i\to j$ and $j\to i$; we'll call the set of these directed edges $\dir$. For each vertex $i$, write $\partial i$ for the set of neighbors of $i$. The central object in our proofs will be the \emph{non-backtracking matrix} associated to $G$; this is a linear operator on $\bbC^{2m}$, which we will think of as the vector space of functions $\dir \to \bbC$. Indexing the standard basis of $\bbC^{2m}$ by the elements of $\dir$,
\[
	B_{i\to j, k\to \ell} = 1 \qquad \text{ if $j = k$ and $i \neq \ell$},
\]
and zero otherwise. True to its name, the powers of $B$ encode walks on $G$ which are forbidden from returning along the same edge that they have just traversed. 

The reader may verify that $B$ is a non-normal operator, and therefore its spectrum is in general a complicated subset of the complex plane. Since its entries are nonnegative, however, we can apply the Perron-Frobenius theorem after carefully analyzing the reducibility and periodicity of $B$. The following result, collating \cite[Corollary 11.12]{terras2010zeta} and \cite[Proposition 3.1]{kotani2000zeta}, characterizes these attributes.
 
\begin{proposition}[Terras, Kotani, Sudana] \label{prop:perron}
  Let $G$ be connected. The spectrum of $B$ depends only the $2$-core of $G$, and once we restrict to this core, $B$ is reducible if and only if $G$ is a cycle. Finally, $B$ has even period if and only if $G$ is bipartite, and odd period $p$ if and only if $G$ is a \emph{subdivision}, e.g. if it is obtained by replacing in a smaller graph $H$ every edge with a path of length $p$.
\end{proposition}

\noindent From the perspective of coloring, bipartite graphs and subdivisions are are uninteresting, and vertices outside the $2$-core cannot impact the chromatic number, so let us assume from this point that $G$ is non-bipartite and non-subdivided, with minimum degree two. 

In this case, the Perron-Frobenius theorem tells us that $\spr B \triangleq \rho \in \Spec B$, and that the corresponding left and right eigenvectors have positive entries; this positivity will be important, and is the reason we stated Proposition \ref{prop:perron} in such detail. 

An invaluable tool for further analyzing the spectral properties of $B$ is a classic result relating its characteristic polynomial to the determinant of a quadratic matrix-valued function involving $A$ and $D$ and due in various forms to Ihara, Bass, and Hashimoto; see \cite{kotani2000zeta,angel2015non,bass1992ihara,hashimoto1989zeta}, to name just a few.

\begin{theorem}[Ihara, Bass, Hashimoto] \label{thm:ihara}
	For any graph $G$,
	\begin{equation*}
		\det(z\1 - B) = (z^2 - 1)^{|E| - |V|}\det(z^2\1 - zA + D - \1).
	\end{equation*}
\end{theorem}

We will refer to the matrix-valued quadratic
\[
	L(z) \triangleq z^2\1 - zA + D - \1	
\]
as the \emph{deformed Laplacian}; note that when evaluated at $z = \pm 1$ it returns the standard and `signless' Laplacians $D \pm A$. The former is always singular, and the latter if and only if $G$ is bipartite, so given our assumptions $B$ has an eigenvalue at $+1$ with multiplicity $|E| - |V| + 1$, and one at $-1$ with multiplicity $|E| - |V|$. The remaining eigenvalues correspond to $z \in \bbC$ for which $L(z)$ is singular. The key lemma for Theorem 2 relates the spectrum of $B$ to the semidefiniteness of $L(z)$ for negative $z$; we first encountered it in \cite[p13]{fan2017well}.

\begin{lemma} \label{lem:L-pos}
	For any lower bound $r \in \bbR$ on the smallest real eigenvalue of $B$, $L(r) \succeq 0$.
\end{lemma}

\begin{proof}
	For $r \in \bbR$, the matrices $L(r)$ are symmetric with real spectrum. When $r \ll 0$, $L(r) \succeq 0$ by a simple diagonal dominance argument. It is a standard result that the eigenvalues of a matrix are continuous functions in its entries, so as we increase $r$, the only way $L(r)$ can fail to be PSD is for one of its eigenvalues to cross zero. However, by Theorem \ref{thm:ihara} $L(r)$ cannot be singular for any real $r$ smaller than the smallest real eigenvalue of $B$.
\end{proof}

% subsection notation_and_preliminary_material (end)

\subsection{Theorem \ref{thm:lower}: The Ihara-Bass Identity and Deformed Laplacian} % (fold)
\label{sub:lower-pf}

Let $P \succeq 0$ be any positive semidefinite matrix. Writing $r_\ast$ for the smallest real eigenvalue of $B$, Lemma \ref{lem:L-pos} implies
\begin{equation} \label{eq:innerprod}
  0 \le \langle P, L(r)\rangle = r^2\Tr P - r\langle P, A \rangle + \langle P,D - \1\rangle.
\end{equation}
for every $r \le r_\ast$. One can check that, subject to the constraint $r\le r_\ast$, this function is minimized at the smaller of $r_\ast$ and $-\sqrt{\langle X, D - \1\rangle}$. As an aside, we've shown:

\begin{lemma} \label{lem:real-ram}
  If $G$ is non-bipartite, and $B$ has no real eigenvalues other than $\pm 1$ and $\rho$, then for any $P \succeq 0$,
  \[
    \langle A,P \rangle \ge -2\Tr P \sqrt{\langle D-\1, P \rangle},
  \]
\end{lemma}

\noindent In the $d$-regular case, Theorem \ref{thm:ihara} implies that a non-bipartite $G$ graph is Ramanujan if and only if $B$ has no real eigenvalues besides $\pm 1$ and $\rho = d-1$, and that this condition implies $\langle A, P \rangle \ge -2\Tr P\sqrt{d-1}$. Thus Lemma \ref{lem:real-ram} suggests that this condition on the spectrum of $B$ may be a natural notion of the Ramanujan property for irregular graphs.

The proof of Theorem \ref{thm:lower} will follow from a stronger result:
\begin{align}
    \chi_v(G) 
    \ge \max_{r < r_\ast} \max_W \frac{-r\langle W, A \rangle}{r^2 + \langle W, D - \1 \rangle} + 1 \qquad \textup{s.t.} \qquad W &\succeq 0 \\
    \Tr W &= 1 \nonumber \\
    W_{i,j} &\ge 0 \text{ for all $(i,j) \in E$} \nonumber
\end{align}
Let $W$ satisfy the three conditions above, and assume that $X \succeq 0$ is the Gram matrix witnessing $\vcn(G) = \kappa$, so that $X$ has ones on its diagonal and $X_{i,j} \ge -(\kappa - 1)^{-1}$ if $(i,j) \in E$. We can set $P = X \odot W$ in \eqref{eq:innerprod}, so that
\[
  0 \le \langle X \odot W, L(r) \rangle =  r^2 + \frac{r}{\kappa - 1}\langle W,A \rangle + \langle W,D - \1\rangle.
\]
To prove Theorem \ref{thm:lower}, set $W_{i,j} = 1/n$.

\begin{remark}
  It is a priori possible that, by carefully tuning $W$, this result could be improved to meet the high-girth limit of the upper bounds in Theorem \ref{thm:upper}. We have observed numerically, however, that this is not the case.
\end{remark}

% subsection proof_of_theorem_2_the_non_backtracking_matrix_and_deformed_laplacian (end)

\subsection{Theorem \ref{thm:upper}: A Non-backtracking Random Walk} % (fold)
\label{sub:upper-pf}

To prove Theorem \ref{thm:upper}, we need to produce unit vectors $\vec v_i$ for every $i \in V$, so that the maximum of $\innerprod{\vec v_i}{\vec v_j}$ over all $(i,j) \in E$ is as negative as possible. Assume that $\girth(G) \ge 2m + 1$, so that in particular if any vertices are at distance less than $m$, they are connected by a unique non-backtracking (and, indeed, self-avoiding) walk of length $\dist(i,j)$. Borrowing an insight of \cite{srivastava2018alon}, we will construct these vectors from non-backtracking random walk on the vertices of $G$. By this we mean a random walk which, started at some vertex $i$, chooses on its first step one of the neighbors of $i$, and on subsequent steps makes only non-backtracking moves. Write $X_s$ for the random variable encoding the position of the walk at time $s$, and $\bbP_i$ for its distribution upon starting the walk at vertex $i$. We will remain for the moment agnostic as to the actual transition probabilities, so that it is clear which portions of the argument depend on them, and which do not.

The $v_i$ will be built as follows: set each to have one coordinate for each $j \in V$, with
\[
	(\vec v_i)_j = \frac{1}{\sqrt{m}}(-1)^{\dist(i,j)}\sqrt{\bbP_i[X_{\dist(i,j)} = j]} \qquad \text{if $1 \le \dist(i,j) \le m$, and zero otherwise.}
\]
We've arranged things so that
\[
	\|\vec v_i\|^2 = \frac{1}{m}\sum_{s\in [m]} \sum_{j : \dist(i,j) = s} \bbP_i[X_s = j] = 1,
\]
since after $s$ steps the walk has probability one of reaching \emph{some} vertex at distance $s$ from its starting point. 

It remains to study the inner products between pairs of vectors at neighboring vertices. For any $(i,j) \in E$, the inner product depends only on vertices at distance less than $m$ from both $i$ and $j$. Because of our initial girth assumption, the depth-$m$ neighborhoods of $i$ and $j$ together form a tree in which every vertex $\ell$ satisifies $|\dist(i,\ell) - \dist(j,\ell)| = 1$, and we can divide this into a portion $L$ of vertices closer to $i$ than $j$, and its counterpart $R$ closer to $j$ than $i$. Let us further segment $L$ into layers $\{i\} = L_0,L_1,...,L_{m-1}$ according to distance from $i$, and similarly for $R$. 

\begin{figure}[h]
\centering
\begin{tikzpicture}
  \tikzstyle{vertex}=[circle,fill=black!25,minimum size=17pt,inner sep=0pt]

  \node[vertex] (i) at (0,0) {$i$};
  \node[vertex] (j) at (1,0) {$j$};
  \node[vertex] (l11) at (-1,.5) {};
  \node[vertex] (l12) at (-1,-.5) {};
  \node[vertex] (l21) at (-2,1) {};
  \node[vertex] (l22) at (-2,0) {};
  \node[vertex] (l23) at (-2,-1) {};  
  \node[vertex] (r11) at (2,.5) {};
  \node[vertex] (r12) at (2,-.5) {};
  \node[vertex] (r21) at (3,1) {};
  \node[vertex] (r22) at (3,-1) {};

  \draw (i) -- (j);
  \draw (l11) -- (i);
  \draw (l12) -- (i);
  \draw (l21) -- (l11);
  \draw (l22) -- (l12);
  \draw (l23) -- (l12);
  \draw (r11) -- (j);
  \draw (r12) -- (j);
  \draw (r21) -- (r11);
  \draw (r22) -- (r12);

  \begin{pgfonlayer}{background}

    \filldraw[red!40!white] [smooth cycle] plot[tension=.2] coordinates {(-1.6,1.5) (-2.4,1.5) (-2.4,-1.5)  (-1.6,-1.5)};

    \node (l2) at (-2,-2) {$L_2$};

    \filldraw[red!40!white] [smooth cycle] plot[tension=.2] coordinates {(-.6,1.5) (-1.4,1.5) (-1.4,-1.5)  (-.6,-1.5)};

    \node (l1) at (-1,-2) {$L_1$};

    \filldraw[red!40!white] [smooth cycle] plot[tension=.2] coordinates {(.4,1.5) (-.4,1.5) (-.4,-1.5)  (.4,-1.5)};

    \node (l0) at (0,-2) {$L_0$};

    \filldraw[blue!40!white] [smooth cycle] plot[tension=.2] coordinates {(1.4,1.5) (.6,1.5) (.6,-1.5)  (1.4,-1.5)};

    \node (r0) at (1,-2) {$R_0$};

    \filldraw[blue!40!white] [smooth cycle] plot[tension=.2] coordinates {(1.6,1.5) (2.4,1.5) (2.4,-1.5)  (1.6,-1.5)};

    \node (r1) at (2,-2) {$R_1$};

    \filldraw[blue!40!white] [smooth cycle] plot[tension=.2] coordinates {(2.6,1.5) (3.4,1.5) (3.4,-1.5)  (2.6,-1.5)};

    \node (r2) at (3,-2) {$R_2$};

  \end{pgfonlayer}
\end{tikzpicture}
\end{figure}

\noindent Then, directly computing,

\begin{align*}
	\innerprod{\vec v_i}{\vec v_j}
	&= \frac{1}{m}\sum_{\ell:\dist(i,\ell),\dist(j,\ell) \in [m]} (-1)^{\dist(i,\ell) + \dist(j,\ell)} \sqrt{\bbP_i[X_{\dist(i,\ell)} = \ell]\bbP_j[X_{\dist(j,\ell)} = j]} \\
	&= \frac{-1}{m}\sum_{s\in[m-1]}\left(\sum_{\ell \in L_s}\sqrt{\bbP_i[X_s = \ell]\bbP_j[X_{s+1} = \ell]} + \sum_{\ell \in R_s} \sqrt{\bbP_i[X_{s+1} = \ell]\bbP_j[X_s = \ell]}\right).
\end{align*}

The non-backtracking structure of the random walk, and the local tree-like configuration nearby $i$ and $j$, allow us to simplify this expression further. When $s\ge 1$ Bayes rule implies
$$
  \bbP_i[X_{s+1} = \ell] = \bbP_j[X_{s+1} = \ell\mid X_1 = i]\bbP[X_1 = i].
$$
Now, by non-backtracking, the probability of reaching $\ell$ in $s+1$ steps starting from $j$, conditional on reaching $i$ on the first step, is the same as the probability of reaching $\ell$ in $s$ steps starting at $i$, conditional on the first step not hitting $j$. We can use Bayes again to write
$$
  \bbP_i[x_s = \ell \mid X_1 \neq j] = \frac{\bbP_i[X_s = \ell, X_1 \neq j]}{\bbP_i[X_1 \neq j]}.
$$
Finally, again by non-backtracking, the information that $X_1 \neq j$ is redundant once we know that it starts at $i$ and reaches $s$ in $\ell$ steps, so $\bbP_i[X_s = \ell, X_1 \neq j] = \bbP_i[X_s = \ell]$. Putting together these steps gives us

\begin{align*}
	\bbP_j[X_{s+1} = \ell]
	% &= \bbP_j[X_{s+1} = \ell | X_1 = i] \bbP_j[X_1 = i] & & \text{Bayes} \\
	% &= \bbP_i[X_s = \ell | X_1 \neq j]\bbP_j[X_1 = i] & & \text{non-backtracking} \\ 
	% &= \frac{\bbP_i[X_s = \ell, X_1 \neq j]}{\bbP_i[X_1 \neq j]}\bbP_j[X_1 = i] & & \text{Bayes} \\
	% &= 
  = \frac{\bbP_i[X_s = \ell]}{\bbP_i[X_1 \neq j]}\bbP_j[X_1 = i],
  % & & \text{non-backtracking}.
\end{align*}

\noindent and thus

\begin{align*}
  \innerprod{\vec v_i}{\vec v_j} 
  &= -\frac{1}{m}\left(\sum_{s\in[m-1]}\left(\sum_{\ell \in L_s} \sqrt{\frac{\bbP_j[X_1 = i]}{\bbP_i[X_1\neq j]}}\bbP_i[X_s = \ell] + \sum_{\ell \in R_s}\sqrt{\frac{\bbP_i[X_1 = j]}{\bbP_j[X_1\neq i]}}\bbP_j[X_s = \ell]\right)\right) \\
  &= -(1-1/m)\left(\sqrt{\bbP_j[X_1 = i]\bbP_i[X_1 \neq j]} + \sqrt{\bbP_i[X_1 = j]\bbP_j[X_1 \neq i]}\right).
\end{align*}

We now choose the transition probabilities for our random walk, having simplified the dependence on them of the inner products we are interested in. Recall from the Perron-Frobenius theorem that, under our assumptions on $G$ (simple, minimum degree $2$, non-subdivided), $\rho$ is a simple eigenvalue of $B$, and that its corresponding left and right eigenvectors have strictly positive entries. Let's denote the right eigenvector by $\vec \phi$, and record explicitly that

\begin{equation} \label{eq:eig}
  \sum_{k \in  \partial j \setminus i} \phi_{j\to k} = \rho \phi_{i \to j} \qquad \forall i\to j \in \dir.
\end{equation}

\noindent It will be useful to overload notation and define $\phi_i \triangleq \sum_{j \in \partial i} \phi_{i \to j}$, observing that \ref{eq:eig} implies $\phi_i = \phi_{i \to j} + \rho \phi_{j \to i}$ for every $j \in \partial i$. 

We will set the transition probabilities of our random walk proportional to the coordinates of $\vec\phi$. In other words, 
\[
  \bbP_i[X_1 = j] = \frac{\phi_{i\to j}}{\phi_i} \qquad \text{if $i\to j \in \dir$}
\]
and
\[
  \bbP_i\left[X_s = \ell \mid X_{s-1} = k, X_{s-2} = j \right] = \frac{\phi_{k\to \ell}}{\rho\phi_{j\to k}} \qquad \text{if $s > 1$ and $j\to k \to \ell$ is non-backtracking}.
\]
Normalization follows immediately from the fact that $\phi$ is a right eigenvector. Returning to the inner product between $v_i$ and $v_j$,

\begin{align*}
	\innerprod{\vec v_i}{\vec v_j} 
	&= -(1-1/m)\left(\sqrt{\bbP_j[X_1 = i]\bbP_i[X_1 \neq j]} + \sqrt{\bbP_i[X_1 = j]\bbP_j[X_1 \neq i]}\right) \\
	&= -(1-1/m)\frac{\sqrt\rho(\phi_{i\to j} + \phi_{j\to i})}{\sqrt{\phi_i \phi_j}} \\ 
	&= -(1-1/m)\frac{\sqrt\rho}{\rho + 1}\frac{\phi_i + \phi_j}{\sqrt{\phi_i\phi_j}} & &\text{from \eqref{eq:eig} and discussion}\\
	&\le -(1-1/m)\frac{2\sqrt\rho}{\rho + 1},
\end{align*}
with the final line following (for instance) the inequality of arithmetic and geometric means.

% subsection non-backtracking walk (end)

\subsection{Theorem \ref{thm:main}} % (fold)
\label{sub:main-pf}

We are now prepared to study the vector chromatic number of $G \sim \calG(n,d/n)$. To prove Theorem \ref{thm:main}, we first need to supply a lower bound on $\vcn(G)$---this will follow immediately from Theorem 2, and an established result on the spectrom of $B$ in the \ER case \cite[Theorem 3]{bordenave-lelarge-massoulie}:

\begin{theorem}(Bordenave, Lelarge, and Massoulie)
  When $G \sim \calG(n,d/n)$, with probability $1 - o_n(1)$, the spectrum of $B$ consists of a Perron eigenvalue at $d \pm o_n(1)$, and remaining eigenvalues of magnitude at most $\sqrt{d} + o_n(1)$.
\end{theorem}

\noindent This result in hand, we know w.h.p. the smallest real eigenvalue of $B$ is no smaller than $-\sqrt d - o_n(1)$, and so Theorem \ref{thm:lower} tells us
\[
  \vcn(G) \ge \frac{d^{3/2}}{2d - 1} + 1 + o_n(1)
\]
w.h.p. as well.

We need to show how to apply Theorem \ref{thm:upper} to bound $\vcn(G)$ from above. It is a standard lemma that for any constant $\gamma$, $\girth(G) \ge \gamma$ with constant probability. On this event, the results of Theorem 6 on the spectrum of $G$ still hold with probability $1 - o_n(1)$, and the average degree of $G$ is still $d \pm o_n(1)$, so we can apply Theorem 1 and deduce that, for any $\epsilon$ and any $d$, 
\[
    \vcn(G) \le \frac{d+1}{2\sqrt d} + 1 + \epsilon
\]
with probability bounded away from zero.

We now employ a martingale technique and combinatorial argument due to a string of papers establishing concentration for the chromatic number of \ER graphs \cite{shamir1987sharp,luczak,achlioptas-moore-reg}, and employed in \cite{banks2019lovasz} for a purpose analogous to ours; the presentation is indebted as well to \cite[Theorem 79]{balachandran}. Set $\kappa > 2$ and define a random variable $\Lambda\subset V$ as the largest set of vertices inducing a subgraph of $G$ with vector chromatic number $\kappa$. By Proposition 1, for any $\epsilon$, if we set $\kappa = \frac{d^{3/2}}{2d - 1} + 1 + \epsilon$ then $|\Lambda| = n$ with probability at least $\mu$, for some $\mu \in (0,1)$. 

Think of the random graph $G$ as being sampled in $n$ steps, where on the $i$th one we decide which of the edges will exist between vertex $i$ and the prior $i-1$. If we call $G_i$ the induced subgraph on vertices $[i]\subset V$, then the the random variables $G_1,...,G_n = G$ induce an increasing sequence of sigma algebras, and the sequence $\expected[|\Lambda| \mid G_i]$ is a martingale. The central claim in every application of this martingale method is that, as at each step we are revealing data about the neighborhood of a single vertex, the conditional expectation of $|\Lambda|$ can change by at most one: once the edges between $i$ and the previous vertices are revealed, we can simply delete $i$ from the graph, and our data about the remaining edges is unchanged.

By Azuma's inequality, then, 
\[
  \bbP\left[|(n-|\Lambda|) - \expected(n-|\Lambda|)| > t\sqrt n\right] \le 2e^{-t^2/2}.
\]
Choosing $t$ so that $2e^{-t^2/2}<\mu$, we immediately have $0 \in (\expected[n-|\Lambda|] - t\sqrt n,\expected[n - |x|] + t\sqrt n)$, and thus $n - |\Lambda| \le 2t\sqrt n$ with probability at least $1 - \mu$.

Now, let $\Upsilon \triangleq V\setminus \Lambda$ be the set of vertices which we \emph{cannot} $\kappa$-vector color. We will show that this set can be expanded to one which induces a three-colorable subgraph of $G$, and whose boundary with the remaining $\kappa$-vector colorable portion of $G$ is an independent set. If there are two vertices $i,j \in \Upsilon$ which are (1) not connected to one another by an edge and (2) are both connected to vertices in $\Upsilon$, form a set $\Upsilon_1 = \Upsilon \cup \{i,j\}$, and repeat this process to produce sets $\Upsilon \subset \Upsilon_1 \subset \cdots \subset \Upsilon_M$ until there are no such vertices to add. The boundary of $\Upsilon_M$ is an independent set (or else our expansion process could have continued for another step). Initially, $\Upsilon$ induces a subgraph with at least $|\Upsilon|/2$ edges (because if there were an isolated vertex, we could easily extend the vector coloring to it), and at each step, $|\Upsilon_t| = 2t + |\Upsilon|$, and $|E(\Upsilon_t)| = 3t + |E(\Upsilon)| \ge 3t + |\Upsilon|/2$. If our process progressed long enough for $|\Upsilon_t| = \alpha n$ for some $\alpha$, we'd have $t = (\alpha n - |\Upsilon|)/2$ and 
$$
  |E(\Upsilon_t)| \ge 3/2 (\alpha n - |\Upsilon|) + |\Upsilon|/2 = 3/2\alpha n - |\Upsilon|.
$$
Since $|\Upsilon| = o(n)$, this means the average degree of the subgraph induced by $\Upsilon_t$ would be $3(1 - o(1))$. A union bound shows, though, that small enough subgraphs of size linear in $n$ w.h.p. do not have average degree this high, so the process must terminate when $|\Upsilon_M| = o(n)$. Applying this union bound again, every subgraph of $|\Upsilon_M|$ must have average degree smaller than three, so $\Upsilon_M$ induces a subgraph with no three-core, and can be colored with three colors. 

We now need to produce a valid vector coloring on the entire graph, exploiting the preceding decomposition of $G$ into a subgraph with $\vcn = \kappa$, one with $\chi = 3$, and a independent set separating them. Call $\{\vec v_i\}_{i \in \Lambda}$ the vector coloring on $\Lambda$, and (perhaps by increasing the ambient dimension) let $\vec w_1,\vec w_2,\vec w_3$ be three unit vectors pointing to the corners of a unilateral triangle, and $\vec\zeta$ be a vector orthogonal to $\vec v_i$ and $\vec w_j$. Writing $\sigma : \Upsilon_M \to [3]$ for a valid three-coloring of $\Upsilon_M$, our vector coloring will be
\begin{align*}
  \vec z_i 
  = \begin{cases} 
    \frac{\sqrt{\kappa^2 - 1}}{\kappa}\vec v_i - \frac{1}{\kappa}\vec \zeta & i \in \Lambda \\
    \vec \zeta & i \in \delta \Upsilon_M \\
    \frac{\sqrt 8}{3}\vec w_{\sigma(i)} - \frac{1}{3}\vec \zeta & i \in \Upsilon_M
  \end{cases}
\end{align*}
One can now directly verify that
\begin{align*}
  \innerprod{\vec z_i}{\vec z_j} 
  &\le
  \begin{cases}
    -\frac{1}{\kappa} & \text{$i$ or $j$ is in $\Lambda$} \\
    -\frac{1}{4} & \text{$i$ or $j$ is in $\Upsilon_M$}.
  \end{cases}
\end{align*}

% subsection proof_of_theorem_re (end)

\subsection{Corollaries} % (fold)
\label{sub:corollaries}

\begin{proof}[Proof of Corollary \ref{cor:maxcut}]
  Our vectors $\vec v_i$ from the proof of Theorem \ref{thm:upper} can be used as input to the Goemans-Williamson rounding algorithm \cite{goemans1995improved} for producing large cuts in $G$. Let $X$ be the Gram matrix of the $\vec v_i$, sample $\vec g \sim \calN(0, X)$, and partition vertices according to the sign of the coordinates of $\vec g$. Calculation of the expected size of such a cut is standard: our vectors $v_i$ have inner product at most $-(1-1/m)\frac{2\sqrt\rho}{\rho + 1}$, so
  \begin{align*}
    \expected |\text{cut}| 
    &= \sum_{(i,j)\in E} \bbP[\text{$g_i$ and $g_j$ have different signs}] \\
    &= \sum_{(i,j)\in E} \frac{1}{\pi}\arccos\innerprod{v_i}{v_j} \\
    &\ge \sum_{(i,j)\in E} \left(\frac{1}{2} - \frac{\innerprod{v_i}{v_j}}{\pi}\right) \\
    &\ge |E|\left(\frac{1}{2} + \frac{1}{\pi}(1 - 1/m)\frac{2\sqrt\rho}{\rho + 1}\right)
  \end{align*}
  In the $\calG(n,d/n)$ case, our martingale calculation guarantees with high probability a vector coloring whose inner products satisfy
  $$
    \innerprod{v_i}{v_j} \le - \frac{2\sqrt d}{(\sqrt d + 1)^2},
  $$
  giving us a cut involving at least
  \[
    |E|\left(\frac{1}{2} + \frac{2}{\pi}\frac{\sqrt d}{(\sqrt d + 1)^2}\right) \approx |E|\left(\frac{1}{2} +0.63662\frac{\sqrt d}{(\sqrt d + 1)^2}\right)
  \]
  edges. One can compare this to a non-algorithmic result of Dembo, Montanari, and Sen \cite{dembo2017extremal} that the actual maximum cut severs
  \[
    \approx|E|\left(\frac{1}{2} + 0.7632 \frac{1}{\sqrt{d}} + o_d(\sqrt d)\right) 
  \]
  edges with high probability.
\end{proof}

\begin{proof}[Proof of Corollary \ref{cor:alon}]

  To prove Theorem 3, it suffices to produce a matrix $X \succeq 0$ with unit trace, and for which $\langle L(z), X \rangle$ is small. Returning to the vectors $\vec v_i$ from the proof of Theorem 1, 
  \[
  	X_{i,j} = \sqrt{\phi_i\phi_j}\innerprod{\vec v_i}{\vec v_j},
  \]
  so that $X_{i,i} = \phi_i$ and $X_{i,j} = -(1 - 1/m)\sqrt\rho (\rho + 1)^{-1}(\phi_i + \phi_j)$ for $(i,j) \in E$. Let us scale $\phi$ so that $\Tr X = \sum_i \phi_i = 1$.

  We will need one additional fact. Writing $d_i$ for the degree of vertex $i$, then from \eqref{eq:eig} and surrounding discussion,

  \begin{align*}
    \sum_i \phi_i d_i
    = \sum_i \sum_{j\in\delta i}(\phi_{i\to j} + \rho\phi_{j\to i})
    = (\rho + 1)\sum_i \phi_i = \rho + 1.
  \end{align*}
  Using this and our calculations from the proof of Theorem \ref{thm:upper},
  \begin{align*}
  	\langle X,L(z) \rangle 
  	&\le  z^2 -(1-1/m)\frac{\sqrt\rho}{\rho + 1}\sum_{(i,j) \in E} (\phi_i + \phi_j) + \sum_i \phi_i(d_i - 1) \\
  	&= z^2 - 2(1-1/m)\sqrt\rho \, z + \rho.
  \end{align*}
  Writing out $L(z) = z^2\1 - zA + D - \1$ and rearranging finishes the proof. Notice also that we've shown Lemma 4 is asymptotically tight on high-girth graphs:
  \[
  	\langle A, X \rangle \ge -2(1-1/m)\sqrt\rho = -2(1-1/m)\sqrt{\langle D - \1, X\rangle}.
  \]

\end{proof}

% corollaries (end)

\section*{Acknowledgements} % (fold)
\label{sec:acknowledgements}

We are grateful to Nikhil Srivastava, Archit Kulkarni, Satyaki Mukherjee for illuminating conversations. J.B. is supported by the NSF Graduate Research Fellowship Program under Grant DGE-1752814; L.T is supported by NSF Grant CCF-1815434.

% section acknowledgements (end)

\printbibliography

\end{document}